\documentclass[11pt, reqno]{amsart}
\usepackage{amssymb}
\usepackage{amsmath}
\usepackage{graphicx,color}
\usepackage{wrapfig,framed}
\usepackage{mathtools}
\usepackage[height=22cm, width=15.7cm, hmarginratio={1:1}]{geometry}
\usepackage{hyperref}
\usepackage{float}

\theoremstyle{plain}
\newtheorem{theorem}{Theorem}
\newtheorem{prop}{Proposition}
\newtheorem{lemma}{Lemma}
\newtheorem{cor}{Corollary}

\theoremstyle{definition}
\newtheorem{defi}{Definition}

\newcommand{\beq}{\begin{equation}}
\newcommand{\eeq}{\end{equation}}
\newcommand{\nn}{\nonumber}

\newcommand{\ZZ}{{\mathbb Z}}
\newcommand{\QQ}{{\mathbb Q}}
\newcommand{\CC}{{\mathbb C}}

\newcommand{\A}{{\mathcal A}}

\newcommand{\tr}{{\rm tr}}

\newcommand{\p}{\partial}

\def\={\; = \;}
\def\+{\; + \;}
\def\:={\, := \, }

\title[Matrix resolvents and Ablowitz--Ladik hierarchy]{Tau-functions for 
the Ablowitz--Ladik hierarchy: the matrix-resolvent method}
\author{Mattia Cafasso}
\address{Mattia Cafasso, Univ Angers, CNRS, LAREMA, SFR MATHSTIC, F-49000 Angers, France}
\email{cafasso@math.univ-angers.fr}
\author{Di Yang}
\address{Di Yang, School of Mathematical Sciences, University of Science and Technology of China, Hefei 230026, P.R.~China}
\email{diyang@ustc.edu.cn}
\begin{document}
\maketitle
\begin{abstract}
We extend the matrix-resolvent method for computing logarithmic derivatives of tau-functions to the Ablowitz--Ladik hierarchy. 
In particular, we derive a formula for the generating series of the logarithmic derivatives 
 of an arbitrary tau-function in terms of matrix resolvents. 
As an application, we provide a way of computing certain integrals over the unitary group. 
\end{abstract}

\tableofcontents

\section{Introduction}

The Ablowitz--Ladik (AL) hierarchy has been introduced in the 1970s as a spatial discretization of the AKNS-D hierarchy \cite{AL1,AL2}. Nowadays, it plays an important role in many different areas of mathematical physics. Without being exhaustive, let us recall that its algebro-geometric solutions have been studied in \cite{GHMT, MEKL}. In \cite{AvM2,AvM3}, the same hierarchy is studied under the name of Toeplitz lattice, and 
its tau functions are used to study combinatorial models and integrals over the unitary group. Tau functions for the Ablowitz-Ladik hiearchy 
appear also as correlation functions in quantum spin chains \cite{IIKS}. In \cite{Brini,BCR,LYZZ}, the AL hierarchy is studied to investigate 
the theory of equivariant Gromov-Witten invariants for the resolved conifold of $\mathbb{P}^1$ under the anti-diagonal action. 
In~\cite{LLQZ} the tri-hamiltonian structure of the AL hierarchy is found.
Relations to melting crystal models are treated in~\cite{T}.\\

Recently, the matrix-resolvent method for studying tau-structures (in the sense of \cite{DZ-norm}) 
of differential integrable systems was introduced in~\cite{BDY16,BDY21,Zhou}.
This method was extended to certain differential-difference integrable systems in~\cite{DY17, DY20, Y}.
Our aim is to further extend the matrix-resolvent method to the study of the tau-structure of the AL hierarchy.

We start by giving an algebraic construction of the hierarchy. Let us denote by~$\A$ the polynomial ring generated by infinitely many indeterminates
\[
\A := \QQ \bigl[q_n, r_n, q_{n\pm 1}, r_{n\pm1}, q_{n\pm 2}, r_{n\pm 2}, \dots\bigr] \,.
\] 
Define the shift operator~$\Lambda:\A\to \A$ as the (invertible) linear operator satisfying  
\begin{align}
& \Lambda (q_{n+i})\=q_{n+i+1}\,, \qquad \Lambda (r_{n+i})\=r_{n+i+1}\,, \qquad i \in \ZZ \,,\nn\\
& \Lambda (fg)\= \Lambda (f) \, \Lambda (g)\,, \quad \forall\,f,g\in \A\,. \nn
\end{align}
The original Lax operator of the AL hierarchy is given by
\beq 
\mathcal{L} \:= \Lambda \+ \mathcal{U}(z) \, ,
\eeq
where $z$ plays the role of the spectral parameter and 
\begin{align}
& \mathcal{U}(z) \:= \left(\begin{array}{cc} z& 0\\ 0 & z^{-1}\end{array}\right) \+ \left(\begin{array}{cc} 0 & r_n\\  q_n & 0\end{array}\right) \, . \label{unold} 
\end{align}

For our purposes, it is more convenient to use another Lax operator \cite{MEKL} (see also~\cite{GHMT})
$$
 L := \Lambda \+ U(\lambda) \,,
$$
where $\lambda$ is the new spectral parameter and
\begin{align}
& U(\lambda) \:= \left(\begin{array}{cc} \lambda & r_n \\ \lambda q_n & 1\end{array}\right) \,. \label{un} 
\end{align}

As already observed in \cite{MEKL} (cf.~also~\cite{KMZ}), the two Lax operators are related through the gauge transformation
\beq
U(\lambda) \= G_{n+1}(z) \, \mathcal{U}(z) \, G_n(z)^{-1}\,, \quad \text{with} \quad G_n(z) \:= z^n \begin{pmatrix} 1 & 0 \\ 0 & z \end{pmatrix}
\eeq
and $\lambda = z^2$.

We now recall the definition of matrix resolvents (cf.~\cite{Dickey}) associated to~$L$ (of course, analogously, one can also define the matrix resolvents associated to~$\mathcal L$).
\begin{defi}
	A matrix $R \in {\rm Mat} \bigl(2,\mathcal{A}[[\lambda]]\bigr)$ is a \emph{matrix resolvent} of $L$  if it satisfies the equation
	\begin{equation}\label{eqres}
		\Lambda \bigl(R(\lambda)\bigr) \, U(\lambda) \,-\, U(\lambda) \, R(\lambda) \= 0 \,.
	\end{equation}
\end{defi}
In the following, we will just speak about matrix resolvents, without further specifying that they are associated to $L$.
\begin{lemma}\label{mrdef12}
There exists a unique matrix resolvent  
$R_-(\lambda) \in {\rm Mat} \bigl(2,\mathcal{A}[[\lambda]]\bigr)$
satisfying the normalization conditions
\begin{align}
& R_-(\lambda) \,-\,  \left(\begin{array}{cc} 0 & r_{n-1} \\ 0 & 1 \end{array}\right) \; \in \;  \lambda \, {\rm Mat} \bigl(2,\mathcal{A}[[\lambda]]\bigr) \, ,  \label{normalize11} \\
& {\rm tr} \, R_-(\lambda) \= 1 \,, \qquad  \det R_-(\lambda) \=  0 \,. \label{normalize12}
\end{align}
Similarly, there exists a unique matrix resolvent
$R_+(\lambda) \in {\rm Mat} \bigl(2,\mathcal{A}[[\lambda^{-1}]]\bigr)$ 
satisfying the normalization conditions
\begin{align}
& R_+(\lambda) 
\,-\,  \left(\begin{array}{cc} 1 & 0 \\ q_{n-1} & 0 \end{array}\right) \; \in \;  \lambda^{-1} \, {\rm Mat} \bigl(2,\mathcal{A}[[\lambda^{-1}]]\bigr) \, ,  \label{normalize21} \\
& {\rm tr} \, R_+(\lambda) \= 1 \,, \qquad  \det R_+(\lambda) \=  0 \,. \label{normalize22}
\end{align}
\end{lemma}
\noindent The proof is given in Section~\ref{section2}. 
The  matrix resolvents $R_{\alpha}(\lambda)$, $\alpha \in \{+,-\}$,  uniquely defined by Lemma \ref{mrdef12}, will be called {\it basic matrix resolvents}.
We remark that the second conditions in~\eqref{normalize12}--\eqref{normalize22} can be equivalently written as
\begin{align}
\label{MRprod1122}  \tr \left(R_+(\lambda)^2\right) \= \tr \left(R_-(\lambda)^2\right) \= 1 \,.
\end{align}

For the reader's convenience, the first two terms of the basic matrix resolvent read as follows: 
\begin{align}
& R_-(\lambda) \= \left(\begin{array}{cc} 0 & r_{n-1}\\ 0 & 1\end{array}\right) \+ 
\lambda \left( \begin{array}{cc} q_{n} r_{n-1} & -q_n r_{n-1}^2+r_{n-2} (1-q_{n-1}r_{n-1}) \\ q_n & -q_n r_{n-1}\end{array}\right) \+ 
\mathcal{O}(\lambda^2) \,,  \nn\\
& R_+(\lambda) \= \left(\begin{array}{cc} 1 & 0\\ q_{n-1} & 0\end{array}\right) 
\+ \frac1{\lambda}\left(\begin{array}{cc} -q_{n-1} r_n & r_n\\  q_{n-2}(1-q_{n-1}r_{n-1})-q_{n-1}^2r_n &  q_{n-1}r_n\end{array}\right)  
\+ \mathcal{O}(\lambda^{-2}) \nn \,.
\end{align}
It will be convenient to define series $a_{\alpha}(\lambda), b_{\alpha}(\lambda), c_{\alpha}(\lambda)$, $\alpha \in \{+,-\}$, and 
an infinite set of polynomials $\{ a_{\alpha,p},b_{\alpha,p},c_{\alpha,p} \in \mathcal{A} \, | \, \alpha = \pm, p \geq 0\}$ via 
\begin{align}
& R_-(\lambda) \= \begin{pmatrix} 0 & 0 \\ 0 & 1 \\ \end{pmatrix} \+ 
\begin{pmatrix} a_{-}(\lambda) & b_{-}(\lambda) \\ c_{-}(\lambda) &  -a_{-}(\lambda) \\ \end{pmatrix} \= 
\begin{pmatrix} 0 & 0 \\ 0 & 1 \\ \end{pmatrix} \+ \sum_{p\geq 0} \lambda^{p} 
\begin{pmatrix} a_{-,p} & b_{-,p} \\ c_{-,p} &  -a_{-,p} \end{pmatrix} \,,  \label{R1expand}\\
& R_+(\lambda) \= \begin{pmatrix} 1 & 0 \\ 0 & 0\end{pmatrix} \+ 
\begin{pmatrix} a_+(\lambda) & b_+(\lambda) \\ c_{+}(\lambda) &  -a_+(\lambda) \end{pmatrix} \=
\begin{pmatrix} 1 & 0\\ 0 & 0 \\ \end{pmatrix} \+ \sum_{p\geq 0} \frac1{\lambda^{p}} 
\begin{pmatrix} a_{+,p} & b_{+,p} \\  c_{+,p} &  -a_{+,p} \end{pmatrix} \,. \label{R2expand}
\end{align}
In particular, $a_{-,0}=a_{+,0}=0$, $b_{-,0}=r_{n-1}$, $c_{-,0}=0$, $b_{+,0}=0$, $c_{+,0}=q_{n-1}$.

We also define matrices $V_{\alpha,p}(\lambda)$, $\alpha=+,-$, $p\geq0$ by
\begin{align}
& V_{-,p}(\lambda) \= \left(\lambda^{-p} R_-(\lambda)\right)_{\leq 0} \,-\, \begin{pmatrix} a_{-,p} & b_{-,p}\\ 0 & 0 \\ \end{pmatrix} \, , \label{vplus}  \\
& V_{+,p}(\lambda) \= \left(\lambda^{p}R_+(\lambda)\right)_{\geq 0} \,-\, \begin{pmatrix} 0 & 0\\ c_{+,p} & -a_{+,p} \end{pmatrix} \, ,  \label{vminus}
\end{align}
where $(\,)_{\leq 0}$ denotes the polynomial part in~$\lambda^{-1}$ of the argument, and $(\,)_{\geq 0}$ the polynomial part in~$\lambda$.\\

We now want to define an infinite set of derivations $\{ D_{\alpha,p}, \alpha = \pm, p \geq 0 \}$ on the ring $\A[[\lambda]]$ which are {\it admissible}, i.e. they commute with the operator $\Lambda$. In order to do so, it suffices to define how $D_{\alpha,p}$ acts on $\lambda$, $q_n$ and $r_n$, and then 
extend its action using the Leibnitz rule and admissibility. In the proposition below, $D_{\alpha,p}$ acts on matrices entrywise. 
\begin{prop}\label{propDAL} 
The equations
\beq\label{DAL}
D_{\alpha,p} \, U \:= 
\Lambda \bigl(V_{\alpha,p}\bigr) \, U \,-\, U  \, V_{\alpha,p} \, , \quad \alpha = \pm\,, ~ p \geq 0
\eeq
uniquely define a set of admissible derivations whose restriction to $\mathcal A$ is well defined and such that 
$D_{\alpha,p} \, \lambda = 0$ for $\alpha = \pm$ and $p \geq 0$.
\end{prop}

\noindent Proposition \ref{propDAL} will be proven in Section~\ref{section2} via 
an explicit computation. The construction of the flows in the way above is known and can also be found 
in e.g.~\cite{GHMT}.

Introduce the {\it loop operators} as follows:
\beq
\nabla_-(\lambda) \:= \sum_{p\geq0} \frac{D_{-,p}}{\lambda^{-p+1}} \,, \quad 
\nabla_+(\lambda) \:= \sum_{p\geq0} \frac{D_{+,p}}{\lambda^{p+1}} \,. 
\eeq
By using a uniqueness argument given in~\cite{Y} (see the Lemma~3 and the Proposition~2 of~\cite{Y}) 
we will prove in Section~\ref{section2} the following lemma.
\begin{lemma}\label{nablaR}
For $\alpha = \pm$, the following equations are verified:
\begin{align}
&\nabla_-(\mu) \left(R_\alpha(\lambda)\right) \=  -\frac{\lambda}{\mu} \, \frac{\left[R_-(\mu),R_\alpha(\lambda)\right]}{\mu-\lambda} \,-\, 
\left[Q_-(\mu),R_\alpha(\lambda)\right]\,, \label{nabla-R}\\
&\nabla_+(\mu) \left(R_\alpha(\lambda)\right) \=  \frac{\left[R_+(\mu),R_\alpha(\lambda)\right]}{\mu-\lambda} 
\,-\, \left[Q_+(\mu),R_\alpha(\lambda)\right]\, , \label{nabla+R}
\end{align}
where $Q_{\pm}(\mu)$ are given by
\begin{align}
Q_{-}(\mu) \:= \frac1{\mu} \left(\begin{array}{cc} a_-(\mu) & b_-(\mu) \\ 0 & 0 \end{array}\right), \quad Q_{+}(\mu) \:= \frac1\mu \left(\begin{array}{cc} 0 & 0\\ c_+(\mu) & -a_+(\mu) \end{array}\right)\,.
\end{align}
\end{lemma}
\noindent Expanding the right-hand sides of~\eqref{nabla-R} and~\eqref{nabla+R} in power series, 
we see that these two equations 
can be reformulated as follows:
\begin{align}\label{equivdr}
D_{\beta,q} \, R_\alpha(\lambda) \= \left[V_{\beta,q}(\lambda),R_\alpha(\lambda)\right]\,, \quad \alpha,\beta = \pm\,, \, p,q \geq 0\,.
\end{align}

By using Lemma~\ref{nablaR} or, equivalently,~\eqref{equivdr},  we will also prove in Section~\ref{section2} the following
\begin{prop}\label{PropCommuting}
	The derivations $D_{\alpha,p},\, \alpha = \pm, p \geq 0$, are mutually commuting. We call~\eqref{DAL} the ``abstract AL Hierarchy'' and $D_{\alpha,p}$ the ``AL derivations''.
\end{prop}

Let us proceed with the definition of the tau-structure for the AL hierarchy. 

\begin{defi} \label{defiOmega}
For any $\alpha,\beta = \pm$ and $p,q \geq 0$, define the polynomial $\Omega_{\alpha,p;\beta,q}\in \A$ by the generating series 
\begin{align}
&\sum_{p,q\geq 0} \Omega_{+,p;+,q} \, \lambda^{-p-1} \mu^{-q-1} \= \frac{{\rm tr} \left(R_+ (\lambda) R_+ (\mu)\right)}{(\lambda-\mu)^2} 
\,-\, \frac{1}{(\lambda-\mu)^2}\,, \label{Omdef1}\\
&\sum_{p,q\geq 0} \Omega_{+,p;-,q} \, \lambda^{-p-1} \mu^{q-1} 
\= - \frac{{\rm tr} \left(R_+ (\lambda) R_- (\mu)\right)}{(\lambda-\mu)^2} \+ \frac1{\lambda^2}\quad (|\lambda|>|\mu|)\,, \label{Omdef2}\\
&\sum_{p,q\geq 0} \Omega_{-,p;+,q} \, \lambda^{p-1} \mu^{-q-1} \= - \frac{{\rm tr} \left(R_- (\lambda) R_+ (\mu)\right)}{(\lambda-\mu)^2} 
\+ \frac1{\mu^2} \quad (|\mu|>|\lambda|)\,, \label{Omdef3}\\
&\sum_{p,q\geq 0} \Omega_{-,p;-,q} \, \lambda^{p-1} \mu^{q-1} \= \frac{{\rm tr} \, \bigl(R_- (\lambda) R_- (\mu)\bigr)}{(\lambda-\mu)^2} 
\,-\, \frac{1}{(\lambda-\mu)^2} \,. \label{Omdef4}
\end{align}
\end{defi}
\noindent 
The fact that equations ~\eqref{Omdef1}--\eqref{Omdef4}  give a proper definition of the polynomials $\Omega_{\alpha,p;\beta,q}$ 
is proven in Section~\ref{sectionAL}. The defining equations~\eqref{Omdef1}--\eqref{Omdef4}  can be written more succinctly as
\begin{align}
& \sum_{p,q\geq 0} \frac{ \Omega_{\alpha,p;\beta,q}} 
{\lambda^{\alpha p+1} \mu^{\beta q+1}} \= 
\alpha\beta \left(
\frac{{\rm tr} \left(R_{\alpha}(\lambda) R_{\beta}(\mu) \right)}
{\left(\lambda-\mu\right)^2} 
 \,-\,  r_{\alpha,\beta}(\lambda,\mu)\right) \,, \label{twopointOm}
\end{align}
where $r_{+,+}(\lambda,\mu)=r_{-,-}(\lambda,\mu)=1/(\lambda-\mu)^2$, $r_{+,-}(\lambda,\mu)=1/\lambda^2$, $r_{-,+}(\lambda,\mu)=1/\mu^2$.

\begin{lemma} \label{taustructure}
For any $\alpha,\beta = \pm$ and $p,q\geq0$, the polynomials $\Omega_{\alpha,p;\beta,q}$ and $a_{\alpha,p} \in \mathcal A$ satisfy the following equations:
\begin{align}
& \Omega_{\alpha,p;\beta,q} \= \Omega_{\beta,q;\alpha,p}  \,, \label{taustruc1} \\
& D_{\gamma,r} \left(\Omega_{\alpha,p;\beta,q}\right) \= D_{\beta,q} \left(\Omega_{\alpha,p;\gamma,r}\right) \,, \label{taustruc2}\\
& (\Lambda-1) \left(\Omega_{\alpha,p;\beta,q}\right) \= \alpha \, D_{\beta,q} \left(a_{\alpha,p}\right) \,. \label{taustruc3}
\end{align}
\end{lemma}
The property~\eqref{taustruc1} obviously follows from~\eqref{Omdef1}--\eqref{Omdef4}. 
The proof for~\eqref{taustruc2}--\eqref{taustruc3} is given in Section~\ref{sectionAL}. Following \cite{DZ-norm}, we call the set of polynomials $$\{ a_{\alpha,p}, \Omega_{\alpha,p;\beta,q},\quad \alpha = \pm, p,q\geq 0 \} \subseteq \mathcal A$$ 
the {\it tau-structure} for the AL hierarchy. For $k\geq 3$, let us also define inductively
\beq
\Omega_{\alpha_{k},p_{k};\dots;\alpha_1,p_1} \:= D_{\alpha_{k},p_{k}}  \Omega_{\alpha_{k-1},p_{k-1};\dots;\alpha_1,p_1} \,.
\eeq
Using Lemma~\ref{taustructure}, it is clear that, for any $k \geq 2$, $\Omega_{\alpha_k,p_k;\dots;\alpha_1,p_1}$ is invariant by permutations of the pairs of indices $(\alpha_i,p_i)$, $i=1,\dots,k$.
The following theorem is the main result of this paper.
\begin{theorem}\label{main1}
For every $k\geq3$, and for any fixed $\alpha_1,\dots,\alpha_k\in\{+,-\}$, 
the generating series of $\Omega_{\alpha_1,p_1;\dots;\alpha_k,p_k}$ has the following expression:
\begin{align}
& \sum_{i_1,\dots,i_k\geq 0} \frac{ \Omega_{\alpha_1,i_1;\dots;\alpha_k,i_k}} 
{\lambda_1^{\alpha_1 i_1+1}\cdots \lambda_k^{\alpha_k i_k+1}} \= 
- \prod_{j=1}^k \alpha_j \sum_{S_k/C_k} 
\frac{{\rm tr} \left(R_{\alpha_{\sigma(1)}}(\lambda_{\sigma(1)}) \cdots R_{\alpha_{\sigma(k)}}(\lambda_{\sigma(k)}) \right)}
{\prod_{j=1}^k \left(\lambda_{\sigma(j)}-\lambda_{\sigma(j+1)}\right)} \,, \label{mainid}
\end{align}
where $C_k$ denotes the cyclic group.
\end{theorem}

If we view $q_n$ and $r_n$ as functions of $n$ and ${\bf s}=(s^{\alpha,p})_{\alpha=\pm,p\geq0}$, and 
replace $D_{\alpha,p}$ in~\eqref{DAL} by $\p/\p s^{\alpha,p}$, we obtain a commuting family of 
evolutionary differential--difference equations, the {\it AL hierarchy}. 
Let $\bigl(q_n=q(n,{\bf s}), r_n=r(n,{\bf s})\bigr)$ be an arbitrary solution to this hierarchy. 
Following the scheme of Dubrovin and Zhang~\cite{DZ-norm}, 
we know from Lemma~\ref{taustructure} that there exists a function $\tau(n,{\bf s})$ satisfying
\begin{align}
&  \frac{\p^2 \log \tau(n,{\bf s})}{\p s^{\alpha,p} \p s^{\beta,q}} \= \Omega_{\alpha,p;\beta,q}(n + 1, {\bf s}) \,,\label{taufunction1}\\
& (\Lambda-1) \left(\frac{\p \log \tau(n,{\bf s})}{\p s^{\alpha,p}}\right) \= \alpha \, a_{\alpha,p}(n + 1, {\bf s}) \,, 
\label{taufunction2}\\
& \frac{\tau(n+1,{\bf s}) \, \tau(n-1,{\bf s})}{\tau(n,{\bf s})^2} \= 1 \,-\, q_{n}({\bf s}) \, r_{n}({\bf s}) \,.\label{taufunction3}
\end{align}
We call $\tau(n,{\bf s})$ {\it the} tau-function of the solution $\bigl(q_n=q(n,{\bf s}), r_n=r(n,{\bf s})\bigr)$, although 
it is uniquely determined by~\eqref{taufunction1}--\eqref{taufunction3} up to multiplying by a factor of the form 
\beq
\exp \Biggl\{c_0 n \+ \sum_{\alpha=\pm\atop p\geq 0} c_{\alpha,p}s^{\alpha,p} \Biggr\} \,.
\eeq
\begin{defi}
Let $\bigl(q_n=q(n,{\bf s}), r_n=r(n,{\bf s})\bigr)$ be an arbitrary solution to the AL hierarchy, and 
$\tau = \tau(n,{\bf s})$ the tau-function of the solution.  Let $\big\{(\alpha_\ell, i_\ell), \ell = 1,\ldots,k \big\}$ be 
an arbitrary set of multi-indices, with $\alpha_j \in \{ + , -\}, \; i_\ell \geq 0$ for every $\ell = 1,\ldots,k$. We define the {\it $k$-point correlation functions} $\langle\langle \tau_{\alpha_1,i_1}\dots\tau_{\alpha_k,i_k}\rangle\rangle(n,{\bf s})$ of the solution by 
\beq
\langle\langle \tau_{\alpha_1,i_1}\dots\tau_{\alpha_k,i_k}\rangle\rangle(n,{\bf s}) 
\:= \frac{\p^k\log\tau(n,{\bf s})}{\p s^{\alpha_1,i_1} \dots \p s^{\alpha_k,i_k}} \,,\quad k\geq1\,.
\eeq
\end{defi}
From the above definitions, it is clear that 
\beq\label{llrrOm}
 \langle\langle \tau_{\alpha_1,i_1}\dots\tau_{\alpha_k,i_k} \rangle\rangle(n,{\bf s}) \= 
\Omega_{\alpha_1,i_1; \dots ;\alpha_k,i_k}(n + 1,{\bf s}) \,, \quad k\geq2\,.
\eeq

From formula~\eqref{twopointOm}, Theorem~\ref{main1}, and~\eqref{llrrOm} we arrive at the following  

\begin{cor}\label{main2}
Let $\bigl(q_n=q(n,{\bf s}), r_n=r(n,{\bf s})\bigr)$ be an arbitrary solution to the AL hierarchy. 
For every $k\geq2$, and for any fixed $\alpha_1,\dots,\alpha_k\in\{+,-\}$, 
the generating series of $k$-point correlation functions of the solution has the following expression:
\begin{align}
& \sum_{i_1,\dots,i_k\geq 0} \frac{\langle\langle \tau_{\alpha_1,i_1}\cdots\tau_{\alpha_k,i_k}\rangle\rangle(n,{\bf s})} 
{\lambda_1^{\alpha_1 i_1+1}\cdots \lambda_k^{\alpha_k i_k+1}} \nn\\
&=\; - \, \prod_{j=1}^k \alpha_j \Biggl(\sum_{S_k/C_k} 
\frac{
{\rm tr} \left(R_{\alpha_{\sigma(1)}}(\lambda_{\sigma(1)},n + 1, {\bf s}) \cdots R_{\alpha_{\sigma(k)}}(\lambda_{\sigma(k)},n + 1, {\bf s}) \right)}
{\prod_{j=1}^k \left(\lambda_{\sigma(j)}-\lambda_{\sigma(j+1)}\right)} 
\,-\, \delta_{k,2} \, r_{\alpha_1,\alpha_2}(\lambda_1,\lambda_2)\Biggr) \,, \label{main2eq}
\end{align}
where $R_\beta(\lambda,n,{\bf s})$ is defined by the substitution of $(q_n=q(n,{\bf s}), r_n=r(n,{\bf s}))$ into $R_\beta(\lambda)$, 
and $r_{\alpha,\beta}(\lambda,\mu)$ are defined right after~\eqref{twopointOm}.
\end{cor}

Generating functions containing both infinitely many 
negative and positive powers of~$\lambda$ appeared in~\cite{GGR, GGR2} in the computations of  
LUE and JUE correlators via the {\it Riemann--Hilbert approach} 
(cf.~also~\cite{BR21, Du2}). \\

The paper is organized as follows. In Section~\ref{section2}, we prove Lemma~\ref{mrdef12}. 
In Section~\ref{sectionAL}, we prove Lemma~\ref{taustructure} and Theorem~\ref{main1}. 
An application related to Toeplitz determinants and integrals over the unitary is given in Section~\ref{section4}.

\medskip

\noindent {\bf Acknowledgements.} 
We are grateful to our teacher and friend Boris Dubrovin, 
who initiated this project.
M.C. acknowledges the support
of the H2020-MSCA-RISE-2017 PROJECT No. 778010 IPaDEGAN and the International Research
Project PIICQ, funded by CNRS. D.Y. acknowledges the support of the National Key R and D Program of China 2020YFA0713100 and NSFC 12061131014.

\section{Matrix-resolvents and the AL hierarchy}\label{section2}

In this section, we prove Lemma~\ref{mrdef12}, Propositions \ref{propDAL} and \ref{PropCommuting}. 

\begin{proof}[Proof of Lemma~\ref{mrdef12}]
As in~\eqref{R1expand}--\eqref{R2expand}, let us write 
\begin{align}
& R_-(\lambda) \= \left(\begin{array}{cc} a_-(\lambda) & b_-(\lambda) \\ c_-(\lambda) & 1-a_-(\lambda)\end{array}\right)   \,,  \nn\\
& R_+(\lambda) \= \left(\begin{array}{cc} 1+a_+(\lambda) & b_+(\lambda)\\ c_+(\lambda) & -a_+(\lambda)\end{array}\right)   \,.\nn
\end{align}
Substituting these expressions into \eqref{eqres}, \eqref{normalize11}, \eqref{normalize12}, 
\eqref{normalize21}, \eqref{normalize22},  we find
\begin{align}
&  (\Lambda-1) a_-(\lambda) \+ q_n \, \Lambda b_-(\lambda) \,-\,  \lambda^{-1} \, r_n \, c_-(\lambda) \= 0 \,, \label{a11}\\
&  r_n \, (\Lambda+1) a_-(\lambda) \+ (\Lambda-\lambda) b_-(\lambda) \= r_n \,, \label{a12}\\  
&  q_n \, (\Lambda+1) a_-(\lambda) \,-\, ( \Lambda-\lambda^{-1}) c_-(\lambda) \= q_n \,, \label{a13}\\
&  (\Lambda-1) a_-(\lambda) \,-\, r_n \, \Lambda c_-(\lambda) \+ \lambda \, q_n \, b_-(\lambda) \= 0 \,, \label{a14}\\
& a_-(\lambda) \,-\, a_-(\lambda)^2 \,-\, b_-(\lambda) \, c_-(\lambda) \= 0 \,, \label{a15}
\end{align}
and
\begin{align}
&  (\Lambda-1) a_+(\lambda) \+ q_n \, \Lambda b_+(\lambda) \,-\,  \lambda^{-1} \, r_n \, c_+(\lambda) \= 0 \,,\\
&  r_n \, (\Lambda+1) a_+(\lambda) \+ (\Lambda-\lambda) b_+(\lambda)   \= - r_n \,,\\ 
&  q_n \, (\Lambda+1) a_+(\lambda)  \,-\, \big( \Lambda-\lambda^{-1}\big) (c_+(\lambda))   \= - q_n\,,\\
&  (\Lambda-1) a_+(\lambda) \,-\, r_n \, \Lambda c_+(\lambda) \+ \lambda \, q_n \, b_+(\lambda)  \= 0 \,,\\
& a_+(\lambda) \+ a_+(\lambda)^2 \+ b_+(\lambda) \, c_+(\lambda) \= 0 \,.
\end{align}
Uniqueness for $R_-(\lambda)$ follows easily from~\eqref{a11}--\eqref{a15} and~\eqref{normalize11}, 
where we note that equations~\eqref{a11}--\eqref{a14} and~\eqref{normalize11} 
already determine $R_-(\lambda)$ up to integration constants, and~\eqref{a15} fixes uniquely the constants; 
the existence follows from a careful verification of the compatibility between~\eqref{a11}--\eqref{a15}.

Similar verifications are true for $R_+(\lambda)$. The lemma is proved. 
\end{proof}

\begin{proof}[Proof of Proposition~\ref{propDAL}]
By a straightforward computation, the $(1,1)$ term of the equation~\eqref{DAL} implies that $D_{\alpha,p} \, \lambda = 0$ for $\alpha = \pm$ and for any $p \geq 0$. Moreover, using the Leibniz rule and equations \eqref{eqres}, \eqref{R1expand}, \eqref{R2expand}, we also obtain 
that the entries $(1,2)$ and $(2,1)$ of ~\eqref{DAL} 
can be equivalently written as
\begin{align}
& D_{+,p} (r_n) \=  r_n \, \Lambda(a_{+,p}) \+ \Lambda(b_{+,p}) \+ \delta_{p,0} \, r_n\,, \label{d2rn} \\
& D_{+,p} (q_n) \= q_n \, \Lambda(a_{+,p}) \,-\, \Lambda(c_{+, p}) \,, \label{d2qn} \\
& D_{-,p} (r_n) \= -\, r_n \, \Lambda(a_{-,p}) \,-\, \Lambda(b_{-,p}) \,, \label{d1rn} \\
& D_{-,p} (q_n) \= - \, q_n \, \Lambda(a_{-,p}) \+ \Lambda(c_{-,p}) \+ \delta_{p,0} \, q_n\,. \label{d1qn}
\end{align}
and these equations show that \eqref{DAL} defines indeed how the derivations $\{ D_{\alpha,p}, \alpha=\pm, p\geq 0 \}$ act on the variables $r_n$ and $q_n$. Since the right hand side of equations \eqref{d2rn}--\eqref{d1qn} do not depend on $\lambda$,  we also proved that we can restrict these derivations to the polynomial algebra~$\mathcal A$.
\end{proof}

Let us proceed and prove Lemma~\ref{nablaR}.
\begin{proof}[Proof of Lemma~\ref{nablaR}] 
We will prove identity~\eqref{nabla+R} with $\alpha=+$. The remaining cases can be proved in a similar way. 
Denote by~$W(\lambda,\mu)$ the left-hand side of this 
identity. Clearly, $W(\lambda,\mu)\in \mathcal{A}\otimes sl_2(\CC)[[\lambda^{-1},\mu^{-1}]]\mu^{-1}$. It follows straightforwardly from~\eqref{eqres} and~\eqref{MRprod1122} that $W(\lambda,\mu)$ satisfies the following two equations:
\begin{align}
& \Lambda(W(\lambda,\mu)) \, U(\lambda) \+ \Lambda(R_+(\lambda)) \, \nabla_+(\mu) (U(\lambda)) \,-\, 
\nabla_+(\mu) (U(\lambda)) \,-\, U(\lambda) \, W(\lambda,\mu) \= 0\,, \label{g12231}\\
& \tr \, \bigl( W(\lambda,\mu) \, R_+(\lambda) \+ R_+(\lambda) \, W(\lambda,\mu) \bigr) \= 0\,. \label{g12232}
\end{align}
One can verify that, if there exists a solution  
to~\eqref{g12231}--\eqref{g12232} belonging to 
$\mathcal{A}\otimes sl_2(\CC)[[\lambda^{-1},\mu^{-1}]]\mu^{-1}$, then it is unique. Identity~\eqref{nabla+R} is then proved by 
verifying that its right-hand side also satisfies ~\eqref{g12231}--\eqref{g12232} and belongs to 
$\mathcal{A}\otimes sl_2(\CC)[[\lambda^{-1},\mu^{-1}]]\mu^{-1}$. 
\end{proof}

Now let us prove Proposition~\ref{PropCommuting}.

\begin{proof}[Proof of Proposition~\ref{PropCommuting}]
The proposition can be proven in two different ways. On the one hand, it can be proven using  equations ~\eqref{d2rn}--\eqref{d1qn} and~\eqref{eqres} 
and straightforward computations.\\
On the other hand, equations \eqref{nabla-R}, \eqref{nabla+R} yield
\begin{align}
& \nabla_+(\mu) (U(\lambda)) \= \frac{\Lambda(R_+(\mu)) U(\lambda)-U(\lambda)R_+(\mu)}{\mu-\lambda} 
\,-\, \Lambda(Q_+(\mu))U(\lambda)+U(\lambda)Q_+(\mu)\,, \label{nablaUp} \\
& \nabla_-(\mu) (U(\lambda)) \= -\frac{\lambda}{\mu}\frac{\Lambda(R_-(\mu)) U(\lambda)-U(\lambda)R_-(\mu)}{\mu-\lambda}
\,-\, \Lambda(Q_-(\mu))U(\lambda)+U(\lambda)Q_-(\mu) \,. \label{nablaUn}
\end{align}
The commutativity between the derivations ${D_{\alpha,p}, \, \alpha = \pm, p \geq 0}$ can then be obtained using~\eqref{nablaUp}--\eqref{nablaUn} and~\eqref{nabla-R}--\eqref{nabla+R}.
\end{proof}

We note that 
it follows from~\eqref{nablaUp}--\eqref{nablaUn} and~\eqref{eqres} that for $\alpha\in\{\pm\}$, 
\begin{align}
& \nabla_\alpha(\mu) (r_n) \= \alpha \, \frac{r_n}{\mu} \, \Lambda(a_{\alpha}(\mu)) \+ \alpha \, 
\frac{1}{\mu} \, \Lambda(b_{\alpha}(\mu)) \+ \delta_{\alpha,+} \, \frac{r_n}{\mu} \,, \label{nablaalpharn} \\
& \nabla_\alpha(\mu) (q_n) \= \alpha \, \frac{q_n}{\mu} \, \Lambda(a_{\alpha}(\mu)) \,-\, \alpha \, \frac1{\mu} \Lambda(c_{\alpha}(\mu)) 
\+ \delta_{\alpha,-} \, \frac{q_n}{\mu} \,. \label{nablaalphaqn} 
\end{align}
and, of course, these equations are equivalent to ~\eqref{d2rn}--\eqref{d1qn}.

For the reader's convenience, we list the first few examples of the flows in the AL hierarchy, obtained by replacing $D_{\alpha,p}$ by $\p/\p s^{\alpha,p}$:
\begin{align}
& \frac{\p r_n}{\p s^{-,0}}  \= -r_n \,,  & \frac{\p q_n}{\p s^{-,0}} & \= q_n \,, & \\
& \frac{\p r_n}{\p s^{+,0}} \= r_n \,,  & \frac{\p q_n}{\p s^{+,0}} & \= -q_n \,, & \\
& \frac{\p r_n}{\p s^{-,1}} \= - \, r_{n-1} \+ r_n \, q_n \, r_{n-1} \,,  & \frac{\p q_n}{\p s^{-,1}} & \= q_{n+1} \,-\, r_n \, q_n \, q_{n+1} \,, & \\
& \frac{\p r_n}{\p s^{+,1}} \= r_{n+1} \,-\, r_n \, q_n \, r_{n+1} \,,  & \frac{\p q_n}{\p s^{+,1}} & \= - \, q_{n-1} \+ r_n \, q_n \, q_{n-1} \,. &
\end{align}
The combination 
\beq
\frac{\p }{\p t} \:= \frac{\p}{\p s^{+,1}} \,-\, \frac{\p }{\p s^{-,1}} \,-\, \frac{\p}{\p s^{+,0}} \+ \frac{\p}{\p s^{-,0}}
\eeq
yields the complexified form of the AL equation
\begin{align}
&\frac{\p r_n}{\p t} \= r_{n+1} \,-\, 2 \, r_n \+ r_{n-1} \,-\, r_n \, q_n \, (r_{n+1}+r_{n-1})\,, \\
&\frac{\p q_n}{\p t} \= -q_{n+1} \+ 2 \, q_n \,-\, q_{n-1} \+ r_n \, q_n \, (q_{n+1}+q_{n-1}) \,.
\end{align}
This is the original AL equation~\cite{AL1,AL2}, which is a discretization of the celebrated 
nonlinear Schr\"odinger equation.

\section{From matrix resolvents to tau-functions}\label{sectionAL}
In this section, we prove Lemma~\ref{taustructure} and Theorem~\ref{main1}. 

Let us first show that the $\Omega_{\alpha,p;\beta,q}$ are well defined by~\eqref{Omdef1}--\eqref{Omdef4}. 
Write
\beq\label{mellitid}
R_\alpha(\mu) \= R_\alpha(\lambda) \+ 
R_\alpha ' (\lambda)(\mu-\lambda) \+ (\mu-\lambda)^2 \p_\lambda 
\left(\frac{R_\alpha(\lambda)-R_\alpha(\mu)}{\lambda-\mu}\right).
\eeq
Using this identity and equality~\eqref{MRprod1122}, 
we find that the right-hand side of~\eqref{Omdef1} equals
\begin{align}
&\frac{{\rm tr} \left(R_+ (\lambda) R_+ (\mu)\right)}{(\lambda-\mu)^2} 
\,-\, \frac{1}{(\lambda-\mu)^2} \nn\\
&\qquad  \= - \frac{{\rm tr} \left (R_+ (\lambda) R_+ ' (\lambda) \right)}{\lambda-\mu} \+ 
{\rm tr} \left(R_+ (\lambda) \p_\lambda 
\left(\frac{R_+(\lambda)-R_+(\mu)}{\lambda-\mu}\right)\right).  \label{mellitproof}
\end{align}
Differentiating equality~\eqref{MRprod1122} with respect to~$\lambda$ we obtain that
 ${\rm tr} \left(R_+ (\lambda) R_+ ' (\lambda)\right)$ vanishes. 
From~\eqref{normalize21} it is clear that the second term of the right-hand side of~\eqref{mellitproof} 
belongs to $\lambda^{-2}\mu^{-1}\mathcal{A}[[\lambda^{-1},\mu^{-1}]]$. 
Since the left-hand side of~\eqref{mellitproof} is symmetric by exchange of $\lambda$ and $\mu$, 
we obtain that the right-hand side of~\eqref{mellitproof} 
belongs to $\lambda^{-2}\mu^{-2}\mathcal{A}[[\lambda^{-1},\mu^{-1}]]$. 
Therefore, we have verified that $\Omega_{+,p;+,q}$ are well defined by~\eqref{Omdef1}, and that $\Omega_{+,p;+,0}$ 
and $\Omega_{+,0;+,q}$ vanishes. 
In a similar way, one can verify that $\Omega_{-,p;-,q}$ are well defined by~\eqref{Omdef4}, 
with vanishing $\Omega_{-,p;-,0}$ and $\Omega_{-,0;-,q}$. 
Using~\eqref{normalize21} and~\eqref{normalize11} one could easily deduce that 
$\Omega_{+,p;-,q}$ are well defined by~\eqref{Omdef2},  
with vanishing $\Omega_{+,0;-,q}$, $\Omega_{+,p;-,0}$. For $\Omega_{-,p;+,q}$, by~\eqref{Omdef2}, 
the proof is the same as~that for $\Omega_{+,p;-,q}$ due to the cyclicity of the trace.

Let us proceed with proving Lemma~\ref{taustructure}.
\begin{proof}[Proof of Lemma~\ref{taustructure}] We have
\begin{align}
& \nabla_+(\nu)\left(\sum_{p,q\geq 0} \frac{\Omega_{\alpha,p;\beta,q}}{\lambda^{\alpha p+1} \mu^{\beta q+1}}\right) \nn\\
& \= 
\alpha\beta \, \frac{{\rm tr} \left(\nabla_+(\nu)(R_\alpha (\lambda)) R_\beta (\mu)+ R_\alpha (\lambda)\nabla_+(\nu)(R_\beta (\mu)) \right)}{(\lambda-\mu)^2} \nn\\
& \= 
\alpha\beta \, \frac{{\rm tr} \left(\left(\frac{[R_+(\nu),R_\alpha(\lambda)]}{\nu-\lambda}-[Q_+(\nu),R_\alpha(\lambda)] \right)R_\beta (\mu) 
+ R_\alpha (\lambda) \left(\frac{[R_+(\nu),R_\beta(\mu)]}{\nu-\mu} - [Q_+(\nu),R_\beta(\mu)]\right) \right)}{(\lambda-\mu)^2} \nn\\
& \=  
\alpha\beta \, \frac{{\rm tr} \left(\left(-\frac{1}{\nu-\lambda}  
+ \frac1{\nu-\mu}\right) R_\alpha(\lambda) [R_+(\nu),R_\beta(\mu)] \right)}{(\lambda-\mu)^2} \nn\\
& \= 
\alpha\beta \, \frac{{\rm tr} \left(R_\alpha (\lambda) [R_+(\nu),R_\beta(\mu)] \right)}{(\lambda-\mu)(\mu-\nu)(\nu-\lambda)} \,.  \label{DO1}
\end{align}
Similarly, 
\begin{align}
& \nabla_-(\nu)\left(\sum_{p,q\geq 0} \frac{\Omega_{\alpha,p;\beta,q}}{\lambda^{\alpha p+1} \mu^{\beta q+1}}\right) \nn\\
& \= 
\alpha\beta \, \frac{{\rm tr} \left(\nabla_-(\nu)(R_\alpha(\lambda)) R_\beta (\mu)+ R_\alpha (\lambda)\nabla_-(\nu)(R_\beta (\mu)) \right)}{(\lambda-\mu)^2} \nn\\
& \= 
\alpha\beta \, \frac{{\rm tr} \left(\left( -\frac{\lambda}{\nu} \frac{[R_-(\nu),R_\alpha(\lambda)]}{\nu-\lambda}-[Q_-(\nu),R_\alpha(\lambda)] \right) R_-(\mu) 
+ R_-(\lambda) \left( -\frac{\mu}{\nu} \frac{[R_-(\nu),R_\beta(\mu)]}{\nu-\mu} - [Q_-(\nu),R_\beta(\mu)]\right) \right)}{(\lambda-\mu)^2} \nn\\
& \= 
\alpha\beta \, \frac{{\rm tr} \left(\left(\frac{\lambda}{\nu} \frac{1}{\nu-\lambda}  
- \frac{\mu}{\nu} \frac1{\nu-\mu}\right) R_\alpha(\lambda) [R_-(\nu),R_\beta(\mu)] \right)}{(\lambda-\mu)^2} \nn\\
& \= -\, \alpha\beta \, 
\frac{{\rm tr} \left(R_\alpha (\lambda) [R_-(\nu),R_\beta(\mu)] \right)}{(\lambda-\mu)(\mu-\nu)(\nu-\lambda)} \,. \label{DO2}
\end{align}
From~\eqref{DO1} and~\eqref{DO2} we see the validity of~\eqref{taustruc2}. 

To show~\eqref{taustruc3}, on the one hand, we have
\begin{align}
& (\Lambda-1) \left(\sum_{p,q\geq 0} \frac{\Omega_{\alpha,p;\beta,q}}{\lambda^{\alpha p+1} \mu^{\beta q+1}}\right)  \nn\\
& \= \alpha\beta \, \frac{{\rm tr} \left(\Lambda(R_\alpha (\lambda)) \Lambda(R_\beta (\mu)) - R_\alpha (\lambda)R_\beta (\mu)\right)}{(\lambda-\mu)^2} \nn\\
& \= \alpha\beta \frac{R_{\alpha,12}(\lambda) R_{\beta,21}(\mu)}{\mu (\lambda-\mu)}- \alpha\beta\frac{R_{\alpha,21}(\lambda) R_{\beta,12}(\mu)}{\lambda (\lambda-\mu)}\,.
\end{align}
where in the last equality we used~\eqref{eqres} and the property ${\rm tr} \left(R(\lambda)\right)=1$. On the other hand, 
for $\beta=+$, using~\eqref{nabla+R} we have 
\begin{align}
\alpha \, \nabla_+(\mu) \left(\frac1{\lambda} \, a_\alpha(\lambda)\right) \= 
\alpha \, \frac{R_{\alpha,12}(\lambda) R_{\beta,21}(\mu)}{\mu (\lambda-\mu)} \,-\,\alpha \, \frac{R_{\alpha,21}(\lambda) R_{\beta,12}(\mu)}{\lambda (\lambda-\mu)}\,;
\end{align}
for $\beta=-$, using~\eqref{nabla-R} we have 
\begin{align}
\alpha \, \nabla_-(\mu) \left(\frac1{\lambda} \, a_\alpha(\lambda)\right) \= - \alpha \, \frac{R_{\alpha,12}(\lambda) R_{\beta,21}(\mu)}{\mu (\lambda-\mu)} \+\alpha \, \frac{R_{\alpha,21}(\lambda) R_{\beta,12}(\mu)}{\lambda (\lambda-\mu)}\,.
\end{align}
Identity~\eqref{taustruc3} is proved. 
\end{proof}
We are ready to prove Theorem~\ref{main1}.
\begin{proof}[Proof of Theorem~\ref{main1}]
Let us prove the theorem by induction.
For $k=3$, the validity of identity~\eqref{mainid} was established in the proof of Lemma~\ref{taustructure}. 
Suppose~\eqref{mainid} is true for $k=m$ $(m\geq3)$; consider now $k=m+1$. For the case that $\alpha_{m+1}=+$, we have
\begin{align}
& - \sum_{i_1,\dots,i_m,i_{m+1}\geq 0} 
\frac{\Omega_{\alpha_1,i_1;\dots;\alpha_m,i_m;+,i_{m+1}}} 
{\lambda_1^{\alpha_1 i_1+1}\cdots\lambda_m^{\alpha_m i_m+1}\lambda_{m+1}^{i_{m+1}+1}} \nn\\
& \= 
\nabla_+(\lambda_{m+1}) \left(\prod_{j=1}^m \alpha_j \sum_{\sigma\in S_m/C_m} 
\frac{{\rm tr} \left(R_{\alpha_{\sigma(1)}}(\lambda_{\sigma(1)}) \cdots R_{\alpha_{\sigma(m)}}(\lambda_{\sigma(m)}) \right)}
{\prod_{j=1}^m \left(\lambda_{\sigma(j)}-\lambda_{\sigma(j+1)}\right)}\right) \nn\\
& \= 
\prod_{j=1}^m \alpha_j \, \sum_{\sigma\in S_m/C_m} 
\frac1{\prod_{j=1}^m \left(\lambda_{\sigma(j)}-\lambda_{\sigma(j+1)}\right)} \sum_{l=1}^m \,
{\rm tr} \, \Bigl(R_{\alpha_{\sigma(1)}}(\lambda_{\sigma(1)}) 
\cdots R_{\alpha_{\sigma(l-1)}}(\lambda_{\sigma(l-1)}) \nn\\
& \qquad\qquad \cdot \nabla_+(\lambda_{m+1})  (R_{\alpha_{\sigma(l)}}(\lambda_{\sigma(l)})) \cdot R_{\alpha_{\sigma(l+1)}}(\lambda_{\sigma(l+1)}) \cdots R_{\alpha_{\sigma(m)}}(\lambda_{\sigma(m)}) \Bigr) \nn\\
& \= 
\prod_{j=1}^m \alpha_j \, \sum_{\sigma\in S_m/C_m} 
\frac1{\prod_{j=1}^m \left(\lambda_{\sigma(j)}-\lambda_{\sigma(j+1)}\right)} \sum_{l=1}^m \, 
{\rm tr} \, \Biggl(R_{\alpha_{\sigma(1)}}(\lambda_{\sigma(1)}) 
\cdots R_{\alpha_{\sigma(l-1)}}(\lambda_{\sigma(l-1)}) \nn\\
& \qquad \cdot \biggl(\frac{[R_+(\lambda_{m+1}),R_{\alpha_{\sigma(l)}}(\lambda_{\sigma(l)})]}{\lambda_{m+1}-\lambda_{\sigma(l)}} - 
[Q_+(\lambda_{m+1}),R_{\alpha_{\sigma(l)}}(\lambda_{\sigma(l)})]
\biggr)  \nn\\
& \qquad \cdot R_{\alpha_{\sigma(l+1)}}(\lambda_{\sigma(l+1)}) \cdots R_{\alpha_{\sigma(m)}}(\lambda_{\sigma(m)}) \Biggr)\,. 
\end{align}
Here in the last equality we used~\eqref{nabla+R}. Simplifying this expression with the help of the cyclic invariance of ${\rm tr}$ we obtain the 
validity of~\eqref{mainid} with $k=m+1$ (see the discussions after (4.23) of~\cite{BDY16}). 
Similarly, for the case $\alpha_{m+1}=-$, we have
\begin{align}
& - \sum_{i_1,\dots,i_m,i_{m+1}\geq 0} 
\frac{\Omega_{\alpha_1,i_1;\dots;\alpha_m,i_m;-,i_{m+1}}} 
{\lambda_1^{\alpha_1 i_1+1}\cdots\lambda_m^{\alpha_m i_m+1}\lambda_{m+1}^{-i_{m+1}+1}} \nn\\
& \= 
\nabla_-(\lambda_{m+1}) \left(\prod_{j=1}^m \alpha_j \sum_{\sigma\in S_m/C_m} 
\frac{{\rm tr} \left(R_{\alpha_{\sigma(1)}}(\lambda_{\sigma(1)}) \cdots R_{\alpha_{\sigma(m)}}(\lambda_{\sigma(m)}) \right)}
{\prod_{j=1}^m \left(\lambda_{\sigma(j)}-\lambda_{\sigma(j+1)}\right)}\right) \nn\\
& \= 
\prod_{j=1}^m \alpha_j \, \sum_{\sigma\in S_m/C_m}  
\frac1{\prod_{j=1}^m \left(\lambda_{\sigma(j)}-\lambda_{\sigma(j+1)}\right)} \sum_{l=1}^m \, 
{\rm tr} \, \Bigl(R_{\alpha_{\sigma(1)}}(\lambda_{\sigma(1)}) 
\cdots R_{\alpha_{\sigma(l-1)}}(\lambda_{\sigma(l-1)}) \nn\\
& \qquad\qquad \cdot \nabla_-(\lambda_{m+1})  (R_{\alpha_{\sigma(l)}}(\lambda_{\sigma(l)})) \cdot R_{\alpha_{\sigma(l+1)}}(\lambda_{\sigma(l+1)}) \cdots R_{\alpha_{\sigma(m)}}(\lambda_{\sigma(m)}) \Bigr) \nn\\
& \= 
\prod_{j=1}^m \alpha_j \, \sum_{\sigma\in S_m/C_m} 
\frac1{\prod_{j=1}^m \left(\lambda_{\sigma(j)}-\lambda_{\sigma(j+1)}\right)} \sum_{l=1}^m \, 
{\rm tr} \, \Biggl(R_{\alpha_{\sigma(1)}}(\lambda_{\sigma(1)}) 
\cdots R_{\alpha_{\sigma(l-1)}}(\lambda_{\sigma(l-1)}) \nn\\
& \qquad \cdot \biggl( -\frac{\lambda_{\sigma(l)}}{\lambda_{m+1}} \, 
\frac{[R_-(\lambda_{m+1}),R_{\alpha_{\sigma(l)}}(\lambda_{\sigma(l)})]}{\lambda_{m+1}-\lambda_{\sigma(l)}} - 
[Q_-(\lambda_{m+1}),R_{\alpha_{\sigma(l)}}(\lambda_{\sigma(l)})]
\biggr)  \nn\\
& \qquad \cdot R_{\alpha_{\sigma(l+1)}}(\lambda_{\sigma(l+1)}) \cdots R_{\alpha_{\sigma(m)}}(\lambda_{\sigma(m)}) \Biggr)\,. 
\end{align}
Here in the last equality we used~\eqref{nabla-R}. Simplifying this expression we obtain the 
validity of~\eqref{mainid} with $k=m+1$. The theorem is proved. 
\end{proof}

For the reader's convenience, let us list 
the first few relations in~\eqref{taufunction1}--\eqref{taufunction2}:
\begin{align}
& \Lambda^{-1} {\frac{\p^2 \log \tau_n}{\p s^{+,1} \p s^{+,1}}}
 \= - q_{n-2} r_n (1- q_{n-1}r_{n-1}) \,, \quad \Lambda^{-1} {\frac{\p^2 \log \tau_n}{\p s^{-,1} \p s^{-,1}}} \=  -q_{n} r_{n-2} (1- q_{n-1}r_{n-1})\,, \label{exampleomega00}\\
& \Lambda^{-1} {\frac{\p^2 \log \tau_n}{\p s^{+,1} \p s^{-,1}}}\= 1-q_{n-1}r_{n-1} \,, \label{exampleomega1}\\
& \Lambda^{-1} (\Lambda-1) \left(\frac{\p \log \tau_n}{\p s^{+,1}}\right) \= -q_{n-1} r_n \,, \quad
\Lambda^{-1} (\Lambda-1) \left(\frac{\p \log \tau_n}{\p s^{-,1}}\right) \= - q_{n} r_{n-1} \,, 
\label{taufunction2002}\\
& \Lambda^{-1} {\frac{\p^2 \log \tau_n}{\p s^{+,2} \p s^{+,2}}} \nn\\
& \quad \= - (1-q_{n-1} r_{n-1}) \Bigl((1-q_{n-2}r_{n-2}) \bigl( q_{n-3}^2r_{n-2}r_n - q_{n-4} r_n (1-q_{n-3}r_{n-3}) \nn\\
&  \qquad  + 2q_{n-3}(q_{n-2}r_{n-1}r_n+q_{n-1}r_n^2-(1-q_nr_n)r_{n+1})\bigr) \nn\\
&  \qquad  - q_{n-2}\bigl( r_n (q_{n-2}r_{n-1}+q_{n-1}r_n)^2 - (1-q_nr_n) (2q_{n-1}r_nr_{n+1} \nn\\
&  \qquad + 2q_{n-2}r_{n-1}r_{n+1} + q_nr_{n+1}^2-r_{n+2}(1-q_{n+1}r_{n+1}))\bigr)\Bigr)\,. 
\end{align}
Here $\tau_n=\tau(n,{\bf s})$.
The above explicit expressions~\eqref{exampleomega00}--\eqref{taufunction2002} 
agree with 
the {\it unsymmetric identities} given by Adler and van Moerbeke~\cite{AvM2}. 
Clearly, formula~\eqref{taufunction1} gives the explicit generating series 
of all unsymmetric identities of Adler--van Moerbeke type.

In the generic situation, we note that the above relations~\eqref{exampleomega00}--\eqref{taufunction2002} yield 
\begin{align}
& \frac{r_{n+1}}{r_n} \= \tau_{n+1}^2 \frac{(\Lambda-1) \left(\frac{\p \log \tau_n}{\p s^{+,1}}\right)}{\tau_{n+2}\tau_n-\tau_{n+1}^2} \= \frac{\tau_{n+2}\tau_n}{\tau_{n+1}^2}  \frac{(\Lambda-1) \left(\frac{\p \log \tau_{n+1}}{\p s^{-,1}}\right)}{\frac{\p^2 \log \tau_{n+1}}{\p s^{-,1} \p s^{-,1} }} \,,  \label{20211105qr1}\\
& \frac{q_{n+1}}{q_n} \= \tau_{n+1}^2 \frac{(\Lambda-1) \left(\frac{\p \log \tau_n}{\p s^{-,1}}\right)}{\tau_{n+2}\tau_n-\tau_{n+1}^2} \= \frac{\tau_{n+2}\tau_n}{\tau_{n+1}^2} 
\frac{(\Lambda-1) \left(\frac{\p \log \tau_{n+1}}{\p s^{+,1}}\right)}{\frac{\p^2 \log \tau_{n+1}}{\p s^{+,1} \p s^{+,1} }} \,, \label{20211105qr2}
\end{align}
which will be applied in the next section for concrete computations. 
From~\eqref{taufunction3} and~\eqref{exampleomega1} we also note that 
$\tau_n$ satisfies the following differential-difference equation:
\beq\label{ZZZZ}
\frac{\tau_{n+1} \tau_{n-1}}{\tau_n^2} \= 
\frac{\p^2 \log \tau_n}{\p s^{+,1} \p s^{-,1}}\,,
\eeq
which is recognized as the 2-Toda equation~\cite{UT} and could also be applied for concrete computations.

\section{Generating series for correlators related to CUE}\label{section4}
In this section, we study an important class of solutions to the AL hierarchy related to unitary integrals. 

Since the dependence of an AL tau-function on $s^{+,0}$, $s^{-,0}$ is trivial, 
from now on we will denote ${\bf s}=\left(s^{\alpha,p}\right)_{\alpha=\pm, p\geq1}$ and consider only the 
non-trivial flows.
For $n\geq1$, let $Z(n, {\bf s})$ be the formal powers series of~${\bf s}$ 
defined via the following Toeplitz determinants:
\beq\label{toedet}
Z(n,{\bf s}) \= \det \left(\psi_{m-\ell}({\bf s})\right)_{\ell,m = 0}^{n-1} \, , 
\eeq
where
\beq
\psi(\zeta;{\bf s}) \=  {\rm e}^{V(\zeta,\zeta^{-1})} \, {\rm e}^{\sum_{p \geq 1} (s^{+,p}\zeta^p 
+ s^{-,p} \zeta^{-p}) } \; =: \; \sum_{k \in \mathbb Z} \psi_k({\bf s}) \, \zeta^k \, .
\eeq
Here $V(\zeta,\zeta^{-1})$ is a given Laurent series with no constant terms and convergent in a neighborhood of~$S^1$.  
Set $Z(0,{\bf s})\equiv1$. 
It is well known (see for instance~\cite{AvM2, BaikDeiftSuidan}) that these Toeplitz determinants can be written, equivalently, 
as matrix integrals on the unitary group $\mathcal U(n)$ with respect to the Haar measure:
\beq\label{unitaryintegral}
 \det \left(\psi_{m-\ell}({\bf s})\right)_{\ell,m = 0}^{n-1} \= \int_{\mathcal U(n)} {\rm e}^{\mathrm{Tr} \left(V(U,U^*) + \sum_{p \geq 1}(s^{+,p}U^p + s^{-,p}{U^*}^p)\right)} \, \mathrm d U \, .
\eeq
It is shown in~\cite{AvM2} (see also~\cite{Hi, PS}) that $Z(n,{\bf s})$ coincides 
with a particular tau-function for the AL hierarchy for $n\geq0$. To use the 
matrix-resolvent method of computing the logarithmic derivatives of~$Z(n,{\bf s})$, one needs 
to solve the initial data of the solution corresponding to~$Z(n,{\bf s})$.
Suppose from~\eqref{toedet} we can get 
$Z\bigl(n,s^{+,1},0,\dots\bigr)$. Then 
by using \eqref{20211105qr1}--\eqref{20211105qr2} one can obtain 
$q_n\bigl(s^{+,1},0,\dots\bigr)$ and $r_n\bigl(s^{+,1},0,\dots\bigr)$ 
and so one gets $(q_n({\bf 0}),r_n({\bf 0}))$.
With the knowledge of this initial value, formula~\eqref{main2eq} 
then leads to a way of computing the logarithmic derivatives 
of $\log Z(n,{\bf s})$ evaluated at ${\bf s}={\bf 0}$.

As an illustration of the above algorithm, let us give explicit computations for a simple example, 
which is given by  
\beq\label{secondexample}
V(\zeta) \= \log \left((1 + Q \zeta) \left(1 + Q \zeta^{-1}\right)\right) \, ,
\eeq
where $0 \leq Q < 1$ is a parameter. The $Q=0$ case reduces to $V\equiv 0$ and~\eqref{unitaryintegral} gives the 
CUE partition function.
Computing by induction the Toeplitz determinants we find that 
\beq \label{zns}
Z\bigl(n,s^{+,1},0,\dots\bigr) \= \sum_{\ell = 0}^{n} \, (s^{+,1})^\ell \, \frac{Q^\ell}{\ell!} \, \frac{1-Q^{2n+2-2\ell}}{1-Q^2} \,.
\eeq
In particular, 
\beq 
Z(n,{\bf 0}) \= \sum_{k = 0}^{n} Q^{2k} \= \frac{1-Q^{2n+2}}{1-Q^2} \,.
\eeq
We then find from~\eqref{zns} the initial value of the solution that corresponds to the partition function 
$Z(n,{\bf s})$ (as a tau-function for the AL hierarchy), as follows:
\beq
q(n,{\bf 0}) \= (-1)^n Q^n \frac{1-Q^2}{1-Q^{2n+2}}\,, \quad r(n,{\bf 0}) \= (-1)^n Q^n \frac{1-Q^2}{1-Q^{2n+2}} \,.
\eeq

By calculating the basic matrix-resolvents we obtain from 
Corollary~\ref{main2} the following proposition.

\begin{cor} 
For every $k\geq2$, and for any fixed $\alpha_1,\dots,\alpha_k\in\{+,-\}$, 
the following formula holds true:
\begin{align} \label{86}
& \sum_{i_1,\dots,i_k\geq 0} \frac{\langle \tau_{\alpha_1,i_1}\cdots\tau_{\alpha_k,i_k}\rangle(n-1)} 
{\lambda_1^{\alpha_1 i_1+1}\cdots \lambda_k^{\alpha_k i_k+1}} \= - \,
\prod_{j=1}^k \alpha_j \, \Biggl(\sum_{S_k/C_k} 
\frac{{\rm tr} \left(R_{\alpha_{\sigma(1)}}(\lambda_{\sigma(1)}) \cdots R_{\alpha_{\sigma(k)}}(\lambda_{\sigma(k)}) \right)}
{\prod_{j=1}^k \left(\lambda_{\sigma(j)}-\lambda_{\sigma(j+1)}\right)} \nn\\
& \qquad\qquad\qquad\qquad\qquad \,-\, \delta_{k,2} \, r_{\alpha_1,\alpha_2}(\lambda_1,\lambda_2)\Biggr) \,,
\end{align}
where $\langle \tau_{\alpha_1,i_1}\cdots\tau_{\alpha_k,i_k}\rangle(n):=\langle\langle \tau_{\alpha_1,i_1}\cdots\tau_{\alpha_k,i_k}\rangle\rangle(n,{\bf 0})$, $r_{\alpha,\beta}(\lambda,\mu)$ are defined right after~\eqref{twopointOm}, and 
$R_\alpha(\lambda)$ is given by \eqref{R1expand}--\eqref{R2expand} with $a_{\alpha,p}, b_{\alpha,p}, c_{\alpha,p}$, $p\geq0$, being explicitly given by 
\begin{align}
& a_{-,p}(n) \= \frac{(-Q)^{2n-p}(1-Q^2)(1-Q^{2p})}{(1-Q^{2n})(1-Q^{2n+2})} \,,\\
& b_{-,p}(n) \= \frac{(-Q)^{n-p-1}(1-Q^2)(1-Q^{2n+2p+2})}{(1-Q^{2n})(1-Q^{2n+2})} \,,\\
& c_{-,p}(n) \= \delta_{p\geq1} \, \frac{(-Q)^{n-p+1}(1-Q^2)(Q^{2p-2}-Q^{2n})}{(1-Q^{2n})(1-Q^{2n+2})} \,,\\
& a_{+,p}(n) \= - \frac{(-Q)^{2n-p}(1-Q^2) (1-Q^{2p})}{(1-Q^{2n})(1-Q^{2n+2})} \,,\\
& b_{+,p}(n) \= \delta_{p\geq1} \, \frac{(-Q)^{n-p+1}(1-Q^2)(Q^{2p-2}-Q^{2n})}{(1-Q^{2n})(1-Q^{2n+2})} \,,\\
& c_{+,p}(n) \=  \frac{(-Q)^{n-p-1}(1-Q^2)(1-Q^{2n+2p+2})}{(1-Q^{2n})(1-Q^{2n+2})}  \,.
\end{align}
\end{cor}

Observe, in particular, that  
$$
Z(n,{\bf 0}) \= \frac{\mathbb P(X \leq n)}{1 - Q^2} \,,
$$
where $X$ is a geometric random variable:
\begin{equation}\label{geomrandvar}
	\mathbb P(X = k) \= (1 - Q^2) \, Q^{2k}\,, \quad k \geq 0\,.
\end{equation}
More generally, consider
$$
V(\zeta;z,z') \= \log\left((1 + Q\zeta)^z (1 + Q \zeta^{-1})^{z'}\right)\,, 
$$	
where $z,z'$ are positive integer parameters. In this case 
$$
Z(n,{\bf 0}) \= \frac{\mathbb P(L \leq n)}{(1 - Q^2)^{zz'}}\,,
$$
where $L$ is a random variable describing the directed last passage percolation in a rectangular $(z \times z')$ lattice in which each site is equipped with a geometric random variable.
More precisely, take a $(z \times z')$ lattice of integer points, and associate to each point an independent and identically distributed random variable $X_{ij}$ of the same law as in \eqref{geomrandvar}. Then
$$
L \;:=\; \max_{\gamma : (1,1) \nearrow (z,z')} L(\gamma) \,; \quad L(\gamma) \= \sum_{(i,j) \in \gamma} X_{i,j}\,,
$$
where the maximum is taken on all the upright paths going from $(1,1)$ to $(z,z')$ (see for instance \cite{BaikDeiftSuidan}). We plan to come back to this issue in a subsequent publication.

\end{document}